\documentclass[onecolumn]{IEEEtran}
\usepackage{amsmath,amsthm,amsfonts,amscd,amssymb,graphicx,graphics}
\usepackage{eucal}     
\usepackage{verbatim}       
\usepackage{makeidx}        
\usepackage{epsfig}             

\usepackage{url,tikz}
\usetikzlibrary{arrows,shapes}

\newcommand{\C}{\mathcal{C}}
\newcommand{\R}{\mathbb{R}}
\newcommand{\Z}{\mathbb{Z}}

\newcommand{\LA}[1]{\Lambda_{#1}}
\newcommand{\mode}[2]{\left [#1\right] \ \textrm{mod} \ {#2}}
\newcommand{\V}{\mathcal{V}}

\newcommand{\Vol}[1]{\textrm{Vol}(#1)}

\newcommand{\G}{\mathbf{G}}

\newcommand{\eq}[1]{\begin{equation}#1 \end{equation}}
\newcommand{\Y}{\mathcal{Y}}
\newcommand{\set}{\mathcal{S}}
\newcommand{\W}{W_{\textrm{max}}}

\newcommand{\ef}{{ \textrm{ \it Eff}}}
\newcommand{\aef}{{ \textrm{\it AEff}}}
\newcommand{\diag}{\mathfrak{D}}
\newcommand{\CC}{\mathfrak{C}}
\newcommand{\CG}{\mathsf{C}}
\newcommand{\Hd}{H_{\textrm{diff}}}

\newtheorem{thm}{Theorem}
\newtheorem{rem}{Remark}
\newtheorem{lem}{Lemma}
\newtheorem{cor}{Corollary}
\newtheorem{defin}{Definition}
\newtheorem{ex}{Example}

\author{Amin Jafarian~\IEEEmembership{Member,~IEEE,} and Sriram Vishwanath~\IEEEmembership{Senior~Member,~IEEE} \\
University of Texas at Austin, Austin, TX 78712\\
E-mails: \{jafarian,sriram\}@ece.utexas.edu}

\title{Achievable Rates for $K$-user Gaussian Interference Channels}
\date{ }
\begin{document}
\maketitle
\begin{abstract}
The aim of this paper is to study the achievable rates for a $K$ user Gaussian interference channels {\em for any SNR} using a combination of lattice and algebraic codes. Lattice codes are first used to transform the Gaussian interference channel (G-IFC) into a discrete input-output noiseless channel, and subsequently algebraic codes are developed to achieve good rates over this new alphabet. In this context, a quantity called {\it efficiency} is introduced which reflects the effectiveness of the algebraic coding strategy. The paper first addresses the problem of finding high efficiency algebraic codes. A combination of these codes with Construction-A lattices is then used to achieve non trivial rates for the original Gaussian interference channel.

\end{abstract}

\section{Introduction}
Interference channels, although introduced many decades ago \cite{shannon62,ahlswede74}, remain one of the important challenges in the domain of multiuser information theory. Although significant progress has been made in the two-user case, such as the two-user ``very strong interference" channel \cite{Carleial1975} and two-user ``strong interference" \cite{Sato1981a} channel, our understanding of the general case is still somewhat limited, with some salient exceptions \cite{Annapureddy2008,Shang2009,Motahari2009a,Etkin2008}. In general, our understanding of both the achievable regions and the outer bounds for general ($K$-user) interference channels (IFCs) needs significant work. The largest section of this body of literature exists for the case of 2-user IFCs. Indeed, it is natural that we better understand 2-user systems before developing an understanding of $K$-user for $K>2$ systems. An achievable region for the 2-user general discrete memoryless IFC is developed in \cite{Han1981} using layered encoding and joint decoding. Also, multiple outer bounds have been developed for this case \cite{Kramer04} (many of which also generalize naturally to the $K$-user case). This Han-Kobayashi achievable region \cite{Han1981}  has been shown to be ``good'' for the 2-user Gaussian interference channel (G-IFC) in multiple cases by clever choices of parameters in the outer bounds \cite{Etkin2008}. Unfortunately, an equivalent body of literature does not exist for $K >2$ G-IFCs.

Here, we delve a little further into the capacity results in \cite{Carleial1975} and \cite{Sato1981a}. Intuitively, for G-IFCs, ``very strong interference" regime is effective when the interference to noise ratio (INR) at each receiver is greater than  the square of its (own) signal to noise ratio (SNR). In this case, the receiver decodes the interference first, eliminates it and then decodes its own message \cite{Sridharan2008}. ``Strong interference" for two-user G-IFCs corresponds to the case where the INR is greater than SNR at each receiver. In \cite{Sato1981a}, it is established that in the strong interference regime, decoding both transmitters' messages simultaneously is the right thing to do at each receiver. However, for G-IFCs with more than two users, such a characterization is not directly applicable. Indeed, there is significant work needed to generalize the results that are considered well-established for two-users G-IFCs to $K>2$-user G-IFCs.

As the exact capacity results are few and far between, there is a large and growing body of literature on the degrees of freedom (DoF) of more-than-two-user G-IFCs using {\em alignment} \cite{Motahari2009,Cadambe2008,BreslerParekhTse}. Our interest in this paper is to move away from high SNR DoF analysis and focus on developing the finite SNR achievability results for these channels. By doing so, we desire to take this body of literature towards obtaining the achievable rate regions that utilize alignment at any SNR, and thus a step closer towards better understanding its capacity limits. To this end, we combine  structured coding strategies (lattices) with algebraic alignment techniques. As an example, we use such a methodology to characterize the capacity of $K$-user Gaussian channels \cite{Sridharan2008}. The main idea in \cite{Sridharan2008} is that each receiver first decodes the {\em sum} of all the interferers, eliminates it from the received signal and then decodes its own signal. We further generalize notion in \cite{Sridharan2008a}, where a layered lattice scheme is used to achieve rates that correspond to a DoF of greater than $1$. However, the scheme in \cite{Sridharan2008a} is not necessarily optimal even in terms of the DoF achieved and thus may not be ``good" at finite SNR as well. This paper aims at presenting improved achievable rates for $K$-user Gaussian interference channels over those in \cite{Sridharan2008a}.

There is limited literature on effective and computable outer bounds for multiuser $K>2$ Gaussian interference channels.  Fortunately, outer bounds are now better understood for the case of the degrees of freedom of this channel \cite{Host-Madsen2005,EO09,Cadambe2008}. In this context, it has been shown that for a general $K$-user interference channel, $\frac{K}{2}$ is an upper bound on the total DoF \cite{Host-Madsen2005}. Multiple results have been presented in recent years showing that this upper bound is achievable. For time/frequency varying channels, \cite{Cadambe2008} shows that $\frac{K}{2}$ is also achievable and therefore is the total degrees of freedom of such channels. For constant channels, \cite{Cadambe2009} presents an achievable scheme that has a non-trivial gap from the upper limit of $\frac{K}{2}$. Recently, \cite{Etkin2009} and \cite{Motahari2009} show the existence of schemes that achieve $\frac{K}{2}$ total degrees of freedom for interference channels where the interference channel gains are  irrational. \cite{Etkin2009} also shows that the total degrees of freedom is bounded away from $\frac{K}{2}$ when the   channel gains are assumed rational. In \cite{Motahari2009}, a coding scheme based on layering is presented that achieves the $K/2$ upper limit on DoF for certain classes of interference channels. Given that we understand DoF limits better than outer bounds on the finite-SNR rate region, we resort to showing that the achievable schemes developed in this paper are ``good'' in terms of the DoF it achieves. Checking if they are good at any finite SNR remains an open problem.



The rest of the paper is organized as follows: The channel model for the $K$-user Gaussian IFC is described in Section \ref{sec:model}. Section \ref{sec:prelim} covers definitions \& notations used in this work. Section \ref{sec:Overview:GIFC} summarizes the main results of this paper.  In Section \ref{sec:ddc}, a connection between the original Gaussian Interference channel(G-IFC) and an equivalent discrete deterministic interference channel (DD-IFC) is built. For the equivalent DD-IFC, Section \ref{sec:ddc} defines and determines ``efficient" codebooks. Subsequently, Section \ref{sec:GIFC} applies these efficient discrete-channel codebooks to the original Gaussian interference channel.  Section \ref{sec:GIFC} is subdivided into two sections. Subsection \ref{sec:IGIFC} provides an achievable scheme for a G-IFC with integer channel gains. Subsection \ref{sec:RGIFC} generalizes this coding scheme to settings with non-integral channel gains. Finally, Section \ref{sec:conc} concludes the paper. 

\section{Channel Model}
\label{sec:model}
In this paper, we study a $K$-user G-IFC where the received signal is expressed as:
\begin{equation}
\label{receiver}
Y=HX+N.
\end{equation}

\hspace{1.8in}
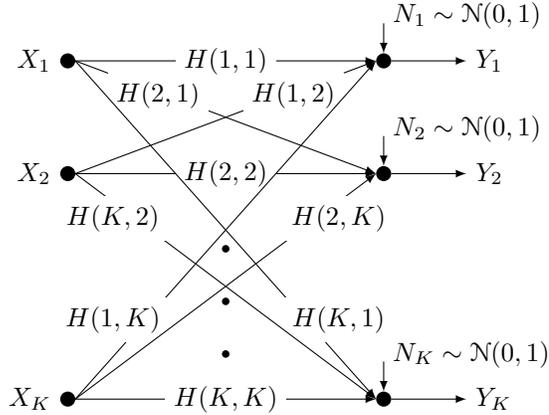
\begin{figure}
\label{fig:model}
\hspace{1.8in}
\begin{tikzpicture}
\draw[-latex] (35mm,15mm)--(75mm,15mm)
node[midway, fill=white] {$H(2,2)$} ;
\draw[-latex] (35mm,30mm)--(75mm,15mm)
node [midway, xshift=-9mm, yshift=3mm, fill=white] {$H(2,1)$};
\draw[-latex] (35mm,30mm)--(75mm,-15mm)
node [midway, xshift=15mm, yshift=-12mm, fill=white] {$H(K,1)$};
\draw (33mm,30mm)
node [left] {$X_1$};
\draw[-latex] (35mm,30mm)--(75mm,30mm)
node[midway, fill=white] {$H(1,1)$} ;
\fill (34mm,30mm) circle (1mm);
\fill (76mm,30mm) circle (1mm);
\draw[-latex] (76mm,35mm)--(76mm,31mm)
node [midway, above right ] {$N_1 \sim \mathcal{N}(0,1)$};
\draw[-latex] (77mm,30mm)--(87mm,30mm)
node[right]{$Y_1$};
\draw[-latex] (35mm,15mm)--(75mm,30mm)
 node [midway, xshift=9mm,yshift=3mm, fill=white]{$H(1,2)$};
\draw[-latex] (35mm,-15mm)--(75mm,30mm)
node [midway, xshift=-15mm,yshift=-12mm, fill=white] {$H(1,K)$};
\draw (33mm,15mm)
node [left] {$X_2$};
\fill (34mm,15mm) circle (1mm);
\fill (76mm,15mm) circle (1mm);
\draw[-latex] (76mm,20mm)--(76mm,16mm)
node [midway, above  right ] {$N_2 \sim \mathcal{N}(0,1)$};
\draw[-latex] (77mm,15mm)--(87mm,15mm)
node[right]{$Y_2$};
\draw[-latex] (35mm,-15mm)--(75mm,15mm)
node [midway, xshift=15mm, yshift=9mm, fill=white] {$H(2,K)$};
\fill (55mm,-2mm) circle (0.5mm);
\fill (55mm,5mm) circle (0.5mm);
\fill (55mm,-9mm) circle (0.5mm);
\draw (33mm,-15mm)
node [left] {$X_K$};
\draw[-latex] (35mm,-15mm)--(75mm,-15mm)
node[midway, fill=white] {$H(K,K)$} ;
\fill (34mm,-15mm) circle (1mm);
\fill (76mm,-15mm) circle (1mm);
\node at  (76,-15mm) [circle,draw] {P};
\draw[-latex] (76mm,-10mm)--(76mm,-14mm)
node [midway, above right ] {$N_K \sim \mathcal{N}(0,1)$};
\draw[-latex] (77mm,-15mm)--(87mm,-15mm)
node[right]{$Y_K$};
\draw[-latex] (35mm,15mm)--(75mm,-15mm)
node [midway,xshift=-15mm, yshift=9mm, fill=white] {$H(K,2)$};

\end{tikzpicture}
\caption{Channel Model}
\end{figure}

In Equation (\ref{receiver}), $X$ denotes the vector-input of size $K$, $Y$ the vector-output of size $K$ and $N$ the vector-noise comprised of independent Gaussian real noise with  power $Z$ and zero mean ($Z$ is assumed to be one through out the paper where it does not hurt the generality of the model). Further, each transmitter satisfies a power constraint, which over $n$ channel uses for User $i$ is given by:

\begin{equation}
\frac{1}{n}\sum_{j=1}^n |X_i(j)|^2 \le P_i.
\end{equation}

In this paper, we focus on a general $K$ user interference model.
With the assumption that $P_i=P$ for all $i$'s, i.e., all the users have the same power constraint. We extend this to a most general setting in the last Section of this work.

Figure \ref{fig:model} shows the channel model. The channel becomes symmetric if we set $H(i,j)=h$ for all $i\neq j$, and $H(i,i)=a$ for all $i$'s.

\section{Preliminaries and Notations}
\label{sec:prelim}
Calligraphic fonts are used to represent sets (such as $\set$).  $|\set|$ represents the cardinality of the set $\set$. Let $r=\mode{a}{b}$ denote the remainder of $a$ when divided by $b$, therefore, $0\le r<a$ and $b$ divides $a-r$; we denote this by $b|a-r$. 

Let $\C_i$ be the codebook at the Transmitter $i$. Assuming that the Transmitter $i$ employs $t$ channel uses to transmit a codeword $X_i \in \C_i$, the rate corresponding to this codebook is given by:
\[R_i=\frac{1}{t}\log(|\C_i|).\] 

The achievable sum rate for a $K$-user G-IFC is defined as: $\displaystyle{\sum_{i=1}^K R_i}.$

Let $P_i$ be the power constraint at Transmitter $i$. The total degrees of freedom is the slope of maximum sum-rate over all coding strategies with respect to $SP=\frac{1}{2}\log\left(\displaystyle{\sum_{i=1}^K P_i}\right)$, as $P_i$ tends to infinity. Formally, we can write it as:
\[TDoF=\lim_{SP \rightarrow \infty}\frac{\displaystyle{\sup_{\textrm{\scriptsize all coding strategies}} \sum_{j=1}^K R_j}}{SP}\]

A lattice is an additive subgroup of $\mathbb{R}^n$ isomorphic to $\Z^n$, where $\mathbb{R}$ and $\Z$ are the set of real and integer numbers. A Construction-A lattice $\Lambda$ is defined as the following \cite{Sloane}:

\[\Lambda=\{x \in \mathbb{R}^n: x=\G z, z\in \Z^n\},\]
where $\G \in \mathbb{R}^{n  \times n}$ is a full rank  $n$ by $n$ matrix.

The Voronoi region of a lattice point $\lambda \in \Lambda$ is defined as all the points in $\mathbb{R}^n$ that are closer to $\lambda$ than any other lattice point. Because of the inherent symmetry  of lattices, we can define the Voronoi region corresponding to zero (which is always a lattice point) as the following:

\[\V=\{x \in \mathbb{R}^n:\|x\|\le \|x-\lambda\| \ \textrm{ \ \ for all \ } \lambda \in \Lambda\}.\]

With a slight abuse of notation, we write: 
\[\mode{a}{\Lambda}=b,\] if and only if $a-b \in \Lambda$.

The second moment per dimension of a lattice is defined as:

\[\sigma^2(\V)=\frac{\int_\V \|x\|^2 dx }{n \Vol{\V}}.\]
 
Let $\mathcal{G}(\Lambda)$ denote the normalized second moment of the lattice $\Lambda$, defined as:
\[\mathcal{G}(\Lambda)=\frac{\sigma^2(\V)}{\Vol{\V}^\frac{2}{n}}.\]
It is known that $\mathcal{G}(\Lambda) \ge \frac{1}{2 \pi e}$  for a general Construction-A lattice \cite{Erez2005}. A lattice is said to be ``good for quantization" if $\mathcal{G}(\Lambda)$ tends to $\frac{1}{2 \pi e}$ as $n$ grows to infinity. Similarly, $\LA{}$ is called ``good for channel coding" if probability of error in decoding a Gaussian noise $Z$ with unit variance from the signal $Y=\lambda+Z$ where $\lambda \in \Lambda$ using lattice decoding (nearest lattice point) goes to zero as $n\rightarrow \infty$ \cite{Erez2005}.

We refer to lattices $\Lambda_c$ and $\Lambda_f$ as a nested pair if $\Lambda_c \subset \Lambda_f$, where subscripts $c$ and $f$ are used to denote the notions of a  ``coarse" and a ``fine" lattice. The nesting ratio of the nested pair $\LA{c} \subset \LA{f}$ is defined as:
\[\rho(\Lambda_c,\Lambda_f) \triangleq \left(\frac{\Vol{\V_c}}{\Vol{\V_f}} \right )^\frac{2}{n}.\]

Note that from how we define ``goodness" of lattices in the above paragraph, if both the lattices $\LA{c}$ and $\LA{f}$ are ``good for quantization", the nesting ratio, $\rho(\Lambda_c,\Lambda_f)$, tends to $\frac{\sigma^2(\V_c)}{\sigma^2(\V_f)}$ for large $n$'s. 

\section{Our Approach: The Central-Dogma}
\label{sec:Overview:GIFC}
Here we describe the essence and intuition behind our approach to obtaining achievable  rates for the $K$ user G-IFC. First, consider those G-IFCs as defined by Equation (\ref{receiver})  that have integer-valued channel gains. In this case, the G-IFC received signal vector is a superposition of integer-scaled values of the transmit signals plus additive noise. To develop an achievable region, we use a two-step coding process as follows. First, each transmitter restricts itself to using a transmit alphabet comprising of elements from the {\em same} ``good" lattice. We call a lattice ``good" if it is good for both quantization and channel coding as defined in the last section.  The channel structure ensures that the receiver observes a valid lattice point, and as the lattice is ``good for channel coding'', the Gaussian noise can be removed at {\em each} of the receivers. We call the resulting noiseless channel a discrete deterministic interference channel (DD-IFC). This DD-IFC is a lattice input lattice output linear transformation channel. Second, we use algebraic alignment mechanisms to determine achievable rates for this DD-IFC. In other words, we determine the largest (symmetric) subset of the ``good" lattice where each receiver can determine its corresponding transmit lattice point.

This two step process is described mathematically as follows:
\vspace{0.3cm}

\noindent{\bf Step 1:} We choose a pair of ``good" n-dimensional nested lattices:
  
\begin{align}
\LA{f}&=\G_1 \left(\frac{1}{q}\G_2.\Z_q+\Z^n \right), \label{eqn:construct1}\\
\LA{c}&=\G_1 \Z^n. \label{eqn:construct2}
\end{align}
where $n$ is a large positive integer, $\G_1 \in \mathbb{R}^{n\times n}$, $\G_2$ is an $n$-dimensional vector from $\Z_q^n$, the operation ``$.$" indicates modulo $q$ multiplication, and $q$ is a prime number. Note that the modulo multiplication can be considered as a real multiplication because of the additive integer part in Equation (\ref{eqn:construct1}).

Note that $\LA{c}\subset\LA{f}$. From \cite{Erez2005}, we know there exist ``good" matrices $\G_1$ and $\G_2$ that render these lattices to be good simultaneously for quantization and channel coding. We choose $\CG_0=\LA{f}\cap \V{c}$ as the transmit alphabet over $n$ channel uses for each transmitter in the system. 
  Given that the channel coefficients are integral, the received signal across $n$ channel uses at each receiver is given by
\[
{\bf Y}_i = \sum_{j=1}^K H(i,j) {\bf X}_j + {\bf N}_i
\]
where ${\bf X}_j$ is an $n$-length transmit lattice point from Transmitter $j$ and ${\bf N}_i$ is the $n$-length Gaussian noise observed at Receiver $i$. Note that $ \sum_{j=1}^K H(i,j) {\bf X}_j$ is also an element of the fine Lattice $\LA{f}$. Since $\LA{f}$  is given to be ``good" for channel coding, Receiver $i$ can, with high probability, determine $ \sum_{j=1}^K H(i,j) {\bf X}_j$ from ${\bf Y}_i$. Thus each receiver can eliminate noise and the system becomes, with high probability, equivalent to the DD-IFC: 

\[
{\bf Y} = H {\bf X} 
\]

where ${\bf X, Y}$ are the $K \times n$ matrices comprised of ${\bf X}_i, {\bf Y}_i~~\forall i$ respectively.
\vspace{0.3cm}

\noindent{\bf Step 2:} Note that finding the sum-rate capacity (and the entire capacity region) of the DD-IFC is a highly non-trivial problem as it is an n-letter channel. Thus, this step focuses on reducing it to a tractable analysis. In order for each receiver to recover its intended message in the DD-IFC, we require algebraic coding to be superposed on the lattice alphabet, i.e., each transmitter uses a (largest possible) subset of transmit alphabet $\CG_0$ that can be decoded at its intended receiver. A closer look at the construction of the fine lattice $\LA{f}$ given in Equation (\ref{eqn:construct1}), shows that every element $a \in \Z_q$ corresponds to a unique point in $\CG_0$. This means the codebook $\C_i$ at Receiver $i$, which is a subset of $\CG_0$, corresponds to a subset  $\hat{\C}_i\in \Z_q$, and can be obtained as following:

\eq{\label{ci:const}\C_i=\left\{\G_1\left(\frac{1}{q} \G_2 \hat{\C}_i + \Z^n \right)\right\}\cap \V_c.}

We transform the problem of construction $\C_i$ from $\CG_0$ at each transmitter into one of determining a {\em one dimensional codebook} $\hat{\C}\in \Z_q$ for the same channel. Let ${\bf E} \in (\hat{\C}_1 \times \hat{\C}_2 \times... \times \hat{\C}_K)$. From Equation (\ref{ci:const}) ${\bf X}_i$ at Transmitter $i$ in the DD-IFC has the following form:
\eq{{\bf X}_i=\frac{1}{q} \G_1\G_2 { E}_i +\G_1 I_i,}
for some $I_i\in \Z^n$, where ${ E}_i$ is the $i^{th}$ element of ${\bf E}$. Thus, we can rewrite ${\bf Y}_i$ at Receiver $i$ as:

\eq{\label{receiver:i} {\bf Y}_i=\sum_{j=1}^K H(i,j)X_j=\frac{1}{q} \G_1\G_2 \sum_{j=1}^K H(i,j) { E}_j +\G_1 I,}
where $I=\displaystyle{\sum_{j=1}^K H(i,j)I_j\in \Z^n}$ because $H$ is an integer matrix.

Assuming that the first entry of the vector $\G_2$ is one, one can obtain $\hat{Y}_i=\displaystyle{\sum_{j=1}^K H(i,j) E_j}$ from  the first entery of the vector $q\G_1^{-1}Y_i \textrm{mod} \ q$, if $\hat{Y}_i<q$ for all $i \in \{1,2,...,K\}$ and $E_i\in \hat{\C}_i$. In other words, we can define an equivalent {\em scalar} DD-IFC  from ${\bf E} \in (\hat{\C}_1\times\hat{\C}_2\times...\times\hat{\C}_K)$ to $\hat{Y}$ as the following:
\eq{\hat{Y}=H {\bf E},}
if the codebook $\hat{\C}_i$'s satisfy the following condition:

\eq{\W=\max\left\{\max_{{\bf E }\in (\hat{\C}_1 \times \hat{\C}_2 \times... \times \hat{\C}_K)} H {\bf E} \right\}+1 \le q,}
where the outer maximization returns the maximum element of the interior vector. 

Thus, with a slight abuse of notation for convenience, we replace the original G-IFC with a scalar DD-IFC given by:

\eq{\label{dd-ifc:model} Y=HX,}

such that

\eq{\label{maxh}\W=\max\left\{\max_{X \in ({\C}_1 \times {\C}_2 \times... \times {\C}_K)} HX \right\}+1 \le q.}

Note that Equation (\ref{maxh}) represents an ``alphabet constraint" for the scalar DD-IFC. This alphabet constraint is more stringent than the one imposed using lattice structure on the G-IFC, but it is sufficient condition that enables us to obtain non-trivial achievable rates for the original G-IFC. For the remainder of this paper, we study the channel definition given by (\ref{dd-ifc:model}) and (\ref{maxh}), and then connect the results obtained with achievable rates for the original Gaussian channel. This concludes the central dogma of our approach in analyzing G-IFCs.

Next, we define the notion of {\em efficiency} for scalar DD-IFCs and connect it with the achievable rate for this channel.  

\begin{defin}
\label{def:ef}
``Efficiency" is defined for a set of  codebooks ${\hat{\C}}=({\hat{\C}}_1,\hat{\C}_2,...,\hat{\C}_K)$ in DD-IFC as the following:
\eq{\label{ddof} \ef(\hat{\C})=\frac{\displaystyle{\sum_{i=1}^K \hat{R}_i}}{\log\left(\W\right)},}
where $\hat{R}_i=\frac{1}{m}\log(|\hat{\C}_i|)$, and $m$ is the number of channel uses Transmitter $i$ utilizes to convey its codeword. 
\end{defin}


In \cite{Jafarian2010}, we establish the following theorem:
\begin{thm}
\label{thm:2010}
The following sum-rate is achievable for a scalar G-IFC (\ref{receiver})
\[\sum_{i=1}^K R_i=\frac{1}{2}\log\left(\frac{P}{N}\right)\times \ef(\hat{\C})\]
\end{thm}

Note that we desire that $\ef(\hat{\C})$ be as close to $K/2$ as possible as efficiency represents the total DoF achieved in the system.
For a symmetric interference channels with integer coefficients, \cite{Cadambe2009} and \cite{Motahari2009} design codes achieving more than one degrees of freedom.
 Theorem \ref{thm:2010} from \cite{Jafarian2010} used ``arithmetic progression codes'' to find efficiencies (much) greater than $1$  for a more general class of interference channels. This results in  ``good" achievable DoF that can be achieved at finite (moderate) SNRs. Although \cite{Jafarian2010} takes a significant step in determining non-trivial rates using this technique, its approach can be generalized and improved further, which is the aim of this paper.

 Before presenting the main result of this paper, we start with the following definitions. Let $\vec{r}$ be a $K$ dimensional vector. Define $\diag (\vec{r})$ to be a $K \times K$ matrix with $\vec{r}$ on its diagonal and zero on non-diagonal entries. 

\begin{rem}
\label{diag_inv}
Let $\vec{d}=(d_1,d_2,...,d_K)$ be an $K$ dimensional vector, $\left(\diag(\vec{d})\right)^{-1}$, is a diagonal matrix with diagonal entries equal to $(\frac{1}{d_1}, \frac{1}{d_2},...,\frac{1}{d_K}).$ With abuse of notations, in this section, we write $\vec{d}^{-1}=(\frac{1}{d_1}, \frac{1}{d_2},...,\frac{1}{d_K})$, and  refer to the matrix $\left(\diag(\vec{d})\right)^{-1}$ as $\diag(\vec{d}^{-1})$.
\end{rem} 

Next, we define the relation between two matrices.

\begin{defin}
\label{relation}
Let $H$ and $H'$ be two integer $K \times K$ matrices. We write $H\sim H'$, $H$ and $H'$ are equivalent, if there exist integer vector $\vec{r}=(r_1,r_2,...,r_K)$ and rational vector $\vec{d}=(d_1,d_2,...,d_K)$, such that:
\[\diag(\vec{d}^{-1})  H \diag(\vec{r})=H'.\]
\end{defin}

 Lemma  \ref{eq_relation} in subsequent sections proves that the defined binary relation, $\sim$  is in fact an equivalence relation. We denote the equivalence class of matrix $H$ as $\CC[H]$. We define $\ef(\CC[H])$ as
\eq{\label{ef_class}\ef(\CC[H])=\sup_{\hat{H} \in \CC[H],  \C_{\hat{H}}} \ef(\C_{\hat{H}}),}
where $\C_{\hat{H}}$ is a  codebook for the DD-IFC channel $\hat{H}$. In particular, we can consider arithmetic progression codebooks that will result in higher efficiencies as demonstrated in Examples \ref{ex1} and \ref{ex2}.

Here we state the main result of this paper for G-IFC's.

\begin{thm}
\label{integer_GIFC}
For the G-IFC channel model in Section \ref{sec:model} with integer channel gains, i.e., $H\in \Z^{K\times K}$, and any $\epsilon>0$, the following sum-rate is achievable:

\[\sum_{i=1}^K R_i=\frac{1}{2}\log\left(\frac{P}{N}\right)\times(\ef(\CC[H])-\epsilon).\]
\end{thm}

We present the proof of the Theorem \ref{integer_GIFC} in the Section \ref{sec:IGIFC}.  
Example \ref{ex2} in Section \ref{sec:ddc} illustrates that, in general, $\ef(\CC[H])>\ef(\C_H)$.



Again, the important point to note is that this is a result for any SNR and is not asymptotic in SNR. An immediate result from Theorem \ref{integer_GIFC} is an achievable total degrees of freedom as stated below:

\begin{cor}
\label{dof}
For the Gaussian interference channel with an integer channel matrix $H$, the following is an achievable total degrees of freedom:
\[DoF=\ef(\CC[H]).\]
\end{cor}

Next, we state a theorem for a general G-IFC without an integer-valued channel matrix. The approach here is to quantize the channel to an integer and employ {\em dithers} to render the error in quantization independent of the desired signal. The scheme for doing so is presented in Section \ref{sec:RGIFC}. 

\begin{thm}
\label{thm_RGIFC}
Let $H \in \R^{K \times K}$ be the channel matrix and $\hat{H}=\lfloor H\rfloor$ be a matrix where each entry is the floor of the corresponding entry in $H$.  Let $\Hd=H-\hat{H}$ be the ``difference" matrix, and
\[H_{dmax}=\max_i\left\{\sum_{j=1}^K \Hd(i,j)^2\right\}.\]
Let $Z_{add}=P H_{dmax}+N$.
For any $\epsilon>0$, the following sum-rate is achievable:

\[\sum_{i=1}^K R_i=\frac{1}{2}\log\left(\frac{P}{Z_{add}}\right) \left(\ef(\CC[\hat{H}])-\epsilon\right).\]
\end{thm}
We prove Theorem \ref{thm_RGIFC} in Section \ref{sec:RGIFC}.

We can further extend our achievable schemes for a more general class of G-IFC. 
Consider a general G-IFC as:
\eq{\label{model2} Y=HX+N,} with power constraint $P_i$ for transmitter $i$, noise variance $N_i$ at receiver $i$ and channel matrix $H$. For a given positive real numbers $P$ and $N$, let $\vec{d}_1(N)=\left(\sqrt{\frac{N_1}{N}},...,\sqrt{\frac{N_K}{N}}\right)$, $\vec{d}_2(P)=\left(\sqrt{\frac{P}{P_1}},...,\sqrt{\frac{P}{P_K}}\right)$, and $H(N,P)$ be defined as:
\[H(N,P)=\diag(\vec{d}_1(N))^{-1}H\diag(\vec{d}_2(P)).\]
It is easy to check that the interference channel model as discussed above is equivalent to an interference channel matrix $H'$ with power constraint $P$ for all the transmitters and noise power $N$ at all the receivers. For this equivalent channel the following remark holds true.
\begin{rem}
\label{thm_IFC}
For all positive $P$ and $N$, Let $\hat{H}(N,P)$, $H(N,P)_{dmax}$ and $Z_{add}$ be defined similar to Theorem \ref{thm_RGIFC} from matrix $H(N,P)$, $P$ and $N$. 
For any $\epsilon>0$, the following sum-rate is achievable for the channel model given by Equation (\ref{model2}):

\[\sum_{i=1}^K R_i=\frac{1}{2}\log\left(\frac{P}{Z_{add}}\right) \left(\ef(\CC[\hat{H}(N,P)])-\epsilon\right).\]
\end{rem}
\begin{proof}
The proof is similar to the proof of Theorem (\ref{thm_RGIFC}). Note that the only condition we use in the proof is the integer channel matrix $\hat{H}$. 
\end{proof}

Note that for the finite SNR, we are able to maximize the achievable sum-rate from Theorem \ref{thm_IFC} on different $P,N>0$. Moreover, if there exists a pair of $P,N>$ such that $H(N,P)\in \Z^{K\times K}$, the following is an achievable total degrees of freedom for this channel:
\[\ef(\CC[\hat{H}(N,P)]).\]


In the next section, we analyze the scalar DD-IFC channel in detail and find achievable efficiencies for an integer-valued channel $H$. We also present some upper-bounds on efficiency for a particular $H$. Note that, in our analysis, we let $q$ grow to infinity. Note that this {\bf does not} correspond to SNR increasing to infinity in the original G-IFC, and that a suitable scaling of lattices always yields a  nested pair at any SNR even though $q \rightarrow \infty$ \cite{Loeliger1997,Erez2005}. 

\section{Efficiency of Scalar DD-IFCs}
\label{sec:ddc}

In this section we use scalar DD-IFCs as defined in the Equation (\ref{dd-ifc:model}). 
Equivalently, one can rewrite  Equation (\ref{dd-ifc:model}) for Receiver $i$ as:
\begin{equation}
\label{dmodel}
Y_i=\sum_{j=1}^K H(i,j) X_j,
\end{equation}
where $X_j \in \Z$. 

Note that, as it is stated in the last section, to maintain tractability in analysis of efficiency for this system, we restrict ourselves to scalar codebooks, i.e., where Transmitter $i$ sends a codeword $X_i \in \C_i$, where $\C_i \subset \Z$.  Receiver $Y_i$ can decode the intended  message from Transmitter $i$ if and only if there exists a function $g_i(.):\Z \rightarrow \C_i$, where: 
\eq{\label{gfunc}g_i(Y_i)=X_i.}

We refer to the set $\C=(\C_1,\C_2,...,\C_K)$ as a codebook if by using sets $\C_i$ at the Transmitter $i$, $Y_i$ can always successfully decode $X_i$ for all $i$.

Define the Set $\set_i$ as 
\[\set_i=\bigoplus_{j \neq i} H(i,j) \cdot \C_j,\]
where $\oplus$ represents the Minkowski sum of sets. It is defined as adding every element of one set to all the element of the other set. The following lemma provides a lower bound on the cardinality of $\set_i$.
\begin{lem}
\label{lowerboundon_set}
\[|\set_i| \ge \sum_{j \neq i} |\C_j| -K+2 \]
\end{lem}
A proof for $K=3$ is given in \cite{HV06}. Here we prove it for a general case:

\begin{proof}
We first prove that $|\C_j \oplus \C_l| \ge |\C_j|+|\C_l|-1$. 
Let $C_{j,M}$ and $C_{l,m}$ be the maximum and the minimum elements of sets $\C_j, \C_l$, respectively.

Let $\set_1=\C_j \oplus C_{l,m}$ and $\set_2=\C_l \oplus C_{j,M}$. Note that  $\set_1 \cap \set_2= \{C_{j,M}+C_{l,m}\}$. Furthermore, $|\set_1|=|\C_j|$, $|\set_2|=|\C_l|$, and $\set_1,\set_2 \subset \C_j\oplus \C_l$. Therefore, $|\C_j\oplus \C_l| \ge |\set_1|+|\set_2|-1= |\C_j|+|\C_l|-1$.

 Using induction and the statement above, one can obtain the desired result.
\end{proof}

Let $\Y_i$ be the set of all  possible received signals $Y_i$ when $X_j \in \C_j$ for all $j$'s. Let $W_i$ be the maximum of $\Y_i$, and $\W=\max_i{W_i}+1$. Note that, another way to define $\W$ is as given in Equation (\ref{maxh}).

Next, we present a lemma that gives us a necessary condition for decodability  at each receiver.
\begin{lem}
\label{DDC_decodability}
There exists a scheme to decode $X_i$ from $Y_i$ or equivalently there exists a function $g(.)$, satisfying Equation (\ref{gfunc}), if and only if the following holds:

\[|\Y_i|=|\C_i||\set_i|.\] 
\end{lem}  
\begin{proof}
We first establish the ``only-if" statement.
Since $\Y_i=\C_i \oplus \set_i$, we know: \eq{\label{ineq1} |\Y_i|\le |\C_i||\set_i| }
We prove the lemma by contradiction. Assume Inequality (\ref{ineq1}) holds strictly, so there are two pairs of $(X_i^1, S_i^1), (X_i^2,S_i^2)$ such that $(X_i^j,S_i^j) \in \C_i \times \set_i$, and:
\[X_i^1+ S_i^1=X_i^2+S_i^2.\]
This clearly implies that Receiver $i$ cannot decode successfully.

The ``if" statement can be proved fairly easily. Assume  $|\Y_i|=|\C_i||\set_i|$, therefore for each $Y_i \in \Y_i$, there exits a unique $(X_i,S_i) \in \C_i \times \set_i$ such that $Y_i=H(i,i) X_i+S_i$. For this choice of $Y_i$, let $g_i(Y_i)=X_i$. This defines the decoding function  $g_i()$.
\end{proof}

Given this framework,  we present a set of achievable schemes for integer-valued scalar DD-IFCs and compare the resulting efficiencies.

The following theorem presents upper and lower bounds on efficiency (as defined in Definition \ref{def:ef}).
\begin{thm}
\label{bounds_ef}
For a channel given by Equation (\ref{dd-ifc:model}) with an integer-valued channel matrix $H$, the following hold true:

\begin{enumerate}
\item  For any $\epsilon>0$, there exists a set of codebooks $\C=(\C_1,\C_2,...,\C_K)$, letting only one user transmits, where \[\ef(\C) > 1-\epsilon.\]
\item For any set of codebooks $\C$ , \[\ef(\C) < \frac{K}{2} \]
\end{enumerate}
\end{thm}

\begin{proof}
\begin{enumerate}
\item 
For the first part, consider the following codebooks:
\begin{align*}
\C_1&=\{0,1,2,...,T\} \ \ \ \textrm{ for some $T>1$}\\
\C_i&=\{0\}  \ \ \ \ \ \forall i, \ 2 \le i \le K
\end{align*}

For Receiver 1, $X_1=\frac{Y_1}{H(i,i)}$ and for all other receivers, $X_i=0$ is the transmitted codeword, and thus it is a valid achievable scheme.
Let $H_m=max(H)$ be the maximum channel gain. One can check that in this case $\W\le H_m \cdot T+1$. For this achievable scheme, efficiency can be computed as the following:
\[\ef(\C)=\frac{\log(T)}{\log(H_m \cdot T+1)}.\]

The efficiency calculated above goes to one as $T$ goes to infinity. Equivalently, we are able to choose a large enough $T_\epsilon$ that the efficiency of the resulted codebook is greater than $1-\epsilon$.

\item
For the upper bound, we assume, without loss of generality, 

\[|\C_1| \ge |\C_2| \ge ... \ge |\C_K|.\]

From Lemma \ref{DDC_decodability}, one can infer that $W_i \ge |\Y_i| = |\C_i| \cdot |\set_i|$. Also using Lemma \ref{lowerboundon_set}, we know that:
\[|\set_i|\ge \sum_{j \neq i} |\C_j|-K+2.\]

Given that $|\C_j|\ge 1$ for all $j$'s,  
\eq{\label{W_upper}\W > W_1 \ge |\C_1| \cdot  |\C_2|,}
where equality holds if $|\C_l|=1$ for all $l\ge 3$.

On the other hand, $\displaystyle{\prod_{i=1}^K |\C_i| \le |\C_1| \cdot |\C_2|^{K-1}}$, with equality only when $|\C_2|=|\C_j|$ for all $j\ge3$. So, one can write:
\begin{align}
\ef(\C) &\le \frac{\log \left(|\C_1| \cdot |\C_2|^{K-1}\right)}{\log\left(|\C_1| \cdot  |\C_2|\right)} \nonumber\\
&=\frac{\log |\C_1|+(K-1)\cdot \log|\C_2|}{\log|\C_1|+\log|\C_2|} \nonumber\\
&\le \frac{K}{2} \label{eq:finalineq},
\end{align}
where (\ref{eq:finalineq}) holds with equality only if $|\C_1|=|\C_2|$. Note that in order to satisfy $\ef(\C)$ to equal $\frac{K}{2}$, all the inequalities in the analysis above must hold with equality or equivalently, we must have that $|\C_i|=1$ for all $ 1\le i \le K$. But for that case, $\ef(\C)=0$. Thus, we conclude that the Inequality (\ref{eq:finalineq}) is always a strict inequality. 
\end{enumerate}
\end{proof}

 
Note that, in the lattice scheme given by Equations (\ref{eqn:construct1}), (\ref{eqn:construct2}), we let $n$ and $q$ goes to infinity such that $\frac{1}{n}\log(q)$ remains constant to keep the rate of lattice codebook constant. This means that for larger $q$'s we need codebooks $(\C_1,\C_2,...,\C_K)$ with  higher (exponentially growing) rates satisfying Inequality (\ref{maxh}). This, however, can be done using layered codebooks. The idea is to construct codebooks with higher rate using a set of primary codebooks $(\C_1,\C_2,...,\C_k)$ multiple times. This scheme can be suboptimal in general, but we show this is enough to achieve non-trivial rates for G-IFC. For the scalar DD-IFC as given in Equations (\ref{dd-ifc:model}) and (\ref{maxh}), we construct the layered codebook $\C^l_i$ from the primary codebook $\C_i$ as the following:
 
\eq{\label{layered_eq} {\C}^l_{i}=\left\{\sum_{v=0}^{l-1} W^v  m_v | m_v \in \C_i \right\},}
for some positive integers $W$ and $l$. We call $W$ the ``bin size" of the layered code. This codebook construction first proposed for the degrees of freedom of symmetric channels in \cite{Motahari2009}. Here we improve it for non-symmetric channels by choosing a more intelligent $W$. 

\begin{defin}
For the layered codebook $\C^l$ we define ``asymptotic efficiency" as the following:
\[\aef(\C)=\lim_{l \rightarrow \infty} \ef(\C^l).\] 
\end{defin}
Note that $\aef(\C)$ is a function of $W$ too, but to simplify the notation we remove $W$ from the expression.
In this section, we show the existence of an appropriate $W$ that renders $\C_i^l$ a decodable code while achieving high $\aef(\C)$. 

\begin{defin}
\label{def:goodness}
A code $\C$ is called a ``good" code if there exists an appropriate bin size $W$ such that the layered code $\C^l$ constructed using the primary codebook $\C$ satisfies the following two conditions:
\begin{enumerate}
\item $\tilde{\C}_l$ is decodable
\item For large enough $l$'s, $\ef(\C^l)>1$, or equivalently:
\[\aef(\C)>1.\]
\end{enumerate}
\end{defin}

Knowing the cardinality of the layered codebooks that is constructed based on a ``good" code is essential. We obtain it in the next lemma.

\begin{lem}
\label{layared_card}
Let $\C$ be a ``good" codebook, and its corresponding layered codebook be $\C^l$. Then $|\C^l|=|\C|^l.$
\end{lem}
\begin{proof}
The result follows from Equation (\ref{layered_eq}) and the decodability condition from Definition \ref{def:goodness}.
\end{proof}

An immediate choice for $W$ in Equation (\ref{layered_eq}) is $\W$ as defined in (\ref{maxh}). One can check that this choice for $W$ makes the resulting layered code $\C^l$ decodable \cite{Motahari2009}. Note that, a choice of $W$ that is less than this value in Equation (\ref{layered_eq}) results in higher $\aef(\C)$.
In this section, we show that there are, in general, better choices for $W$ than $\W$.  This fact can better understood from the following example.

\begin{ex}
\label{ex1}
Consider the following channel matrix:
\[H=\left[\begin{array}{ccc}
1 &4&3\\
2&1&3\\
6&2&1\end{array}
\right].\]

It is easy to check that the following  codes are decodable for this channel:
\begin{align*}
&\C_1=\{0,1,2,3,4,5\},\\
&\C_2=\{0,3\},\\
&\C_3=\{0,2,4\}.
\end{align*}

From the Equation (\ref{maxh}), we observe that $\W=40$. Let $W=\W$  in the layered code given in Equation (\ref{layered_eq}). We can compute the asymptotic efficiency as:
\[\aef(\C)=\frac{\log(36)}{\log(41)}\]

However, in Example \ref{ex2}, we show that using $W=30$ the resulting layered codebook is still decodable and we can achieve the following asymptotic efficiency:
\[\aef(\C)=\frac{\log(36)}{\log(30)}\]

Note that, although we start from a code with $\ef(\C)<1$,  a clever choice of $W$ in the layered coding scheme, results in another code $(\C_l)$ with efficiency greater than one for large enough $l$'s.  
\end{ex}

In order to present our main result for this new and more intelligently picked $W$, we must go back to the relation  in Definition \ref{relation} between matrices. First, we prove this relation is in fact an equivalence relation. 

\begin{lem}
\label{eq_relation}
The binary relation $\sim$ defined as in Definition \ref{relation} is an equivalence relation.
\end{lem}

\begin{proof}
In order to prove this lemma, we need to show the following three properties of the relation $\sim$:

\begin{enumerate}
\item Reflexivity: $H\sim H$.
\item Symmetry: if $H\sim H'$, then $H' \sim H$. 
\item Transitivity: if $H \sim H'$ and $H' \sim H''$, then $H \sim H''$.
\end{enumerate}

First property follows by assigning $\vec{r}=\vec{d}=(1,1,...,1)$. 

For the second property, let $\diag(\vec{d}^{-1}) H \diag(\vec{r})=H'$. Let:
\begin{align*}
&\vec{r'}=\prod_{i=1}^K r_i \times \left(\frac{1}{r_1}, \frac{1}{r_2},...,\frac{1}{r_K}\right),\\
&\vec{d'}=\prod_{i=1}^K r_i \times \left(\frac{1}{d_1},\frac{1}{d_2},...,\frac{1}{d_K}\right).
\end{align*}
Note that $\vec{r'}$ is an integer vector and $\vec{d'}$ is a rational vector. We can write:

\begin{align*}
\diag(\vec{d'}^{-1}) H' \diag(\vec{r'})&=\left(\frac{1}{\displaystyle{\prod_{i=1}^K r_i}} \diag(\vec{d}) \right)\left[\diag(\vec{d}^{-1}) H \diag(\vec{r}) \right] \left(\diag(\vec{r}^{-1}) \prod_{i=1}^K r_i \right)=H,
\end{align*}
where the last steps follows from Remark \ref{diag_inv}.

For the third property, let $\diag(\vec{d}^{-1}) H \diag(\vec{r})=H'$ and $\diag(\vec{d'}^{-1}) H' \diag(\vec{r'})=H$. Thus,
\eq{\label{trans}
\diag(\vec{d'}^{-1}) \diag(\vec{d}^{-1}) H \diag(\vec{r})  \diag(\vec{r'})=H''.}

Let $\vec{r''}=(r_1 r'_1, r_2 r_2',...,r_K r_K')$ and $\vec{d''}=(d_1 d_1', d_2 d_2'',..., d_K d_K')$. It is easy to check that:
\begin{align*}
&\diag(\vec{d}^{-1}) \diag(\vec{d'}^{-1})=\diag(\vec{d''}^{-1}),\\
&\diag(\vec{r}) \diag(\vec{r'})=\diag(\vec{r''})
\end{align*}

This in addition to Equation (\ref{trans}) proves transitivity property.
This concludes the lemma.
\end{proof}

As a result of Lemma \ref{eq_relation}, we can partition the set of all integer matrices $\Z^{K\times K}$ into different classes.
We denote the class of matrices that includes matrix $H$, by $\CC[H]$. Therefore, $H\sim H'$ if and only if $H'\in \CC[H]$. Equivalently, we can say $H \sim H'$ if and only if $\CC[H']=\CC[H]$.


The following theorem uses a good choice for $W$ to achieve a non-trivial efficiency.
\begin{thm}
\label{eq_class}
Let $H$ be an integer matrix and $\tilde{H}\in \CC[H]$. Let $\C$ be a  codebook for $H$ with efficiency $\ef(\C)$. There exists a layered decodable code $\tilde{\C}^l$ for channel matrix $\tilde{H}$ satisfying:
\[\aef(\tilde{\C})=\ef(\C).\]
\end{thm}
\begin{cor}
\label{w=wmax}
Let $\C$ be a codebook for integer channel $H$ with efficiency $\ef(\C)$. There exists a layered codebook $\bar{C}^l$ for $H$ with the following asymptotic efficiency:
\[\aef(\bar{C})=\ef(\C).\]
\end{cor}
\begin{proof}
This  result is  immediate from Theorem \ref{eq_class} by considering $\tilde{H}=H$. 
\end{proof}

In general, we may be able to achieve higher efficiencies than that stated in Corollary \ref{w=wmax} by considering another $H' \in \CC[H]$ (and therefore $H\in\CC[H']$) and searching for a codebook $\hat{\C}$ for $H'$ such that $\ef(\hat{\C})>\ef(\C)$. Thus, we know from Theorem \ref{eq_class} that there is a layered codebook that achieves an efficiency of $\ef(\hat{\C})$ for $H$. To make this point clear we consider the setting in Example \ref{ex1} again as Example \ref{ex2} bellow. 
\begin{ex}
\label{ex2}
Consider the matrix $H$ defined in Example \ref{ex1}. Let
\[H'=\left[\begin{array}{ccc}
1 &12&6\\
2&3&6\\
3&3&1\end{array}
\right].\]

One can check that $H' \in \CC[H]$, or in other words, $H \sim H'$. Also, the following sets represent  codebooks for this channel:

\begin{align*}
&\C_1=\{0,1,2,3,4,5\},\\
&\C_2=\{0,1\},\\
&\C_3=\{0,1,2\}.
\end{align*}

One can check that, for this codebook, $\W '=30$, and therefore:
\[\ef(\C ')=\frac{\log(36)}{\log(30)}.\]

From Theorem \ref{eq_class}, there exists a layered code $\tilde{\C}$ for channel $H$, such that:
\[\aef([\tilde{\C]})=\ef(\C ').\]
\end{ex}
Next, we prove Theorem \ref{eq_class}.
\begin{proof}[Proof of Theorem \ref{eq_class}]
Let $s_i=|\C_i|$ be the codebook size of the codebook $\C_i$ respectively.
Let $H=\diag(\vec{f}^{-1}) \tilde{H} \diag(\vec{t})$.
Consider the following set of codebooks:

\eq{\label{ci:apc}\tilde{\C}_i=t_i \C_i.}

Next, we prove this code is decodable for channel $\tilde{H}$. Let $\tilde{Y}=(\tilde{Y}_1,\tilde{Y}_2,...,\tilde{Y}_K)$ be the received signal, and $\tilde{M}=(\tilde{m}_1,\tilde{m}_2,...,\tilde{m}_K) \in \tilde{\C}$ be the transmitted signal vector. Let $M=(m_1,m_2,...,m_K)$ be the equivalent vector of $\tilde{M}$ in $\C$, i.e., $m_i=\frac{\tilde{m}_i}{t_i}$ and $Y=HM$ be the output vector of the channel $H$. One can write:
\eq{\label{m_m'}\tilde{M}=\diag(\vec{t})M.} 
We have:
\begin{align}
\tilde{Y}&=\tilde{H} \tilde{M} \notag\\
&=\tilde{H} \diag(\vec{t}) M\notag\\
&=\diag(\vec{f}) H M \notag\\
&=\diag(\vec{f}) Y. \label{y'_y}
\end{align}
 
Thus, to decode $\tilde{M}$, we  first compute 
\eq{\label{y_y'} Y=\diag(\vec{f}^{-1})\tilde{Y}.} 
Then, we decode $M$ and compute $\tilde{M}$ from $M$ using Equation (\ref{m_m'}). Let the maximum channel output of the channel $Y$ when code $\C$ is used at the transmitters be $\W$, i.e.:
\[\W=\max_{\C} Y.\]

Now, we construct the following layered code for the channel $\tilde{H}$:
\eq{\label{layered_cb}\tilde{\C}^l_i=\left\{\sum_{v=0}^{l-1} (\W)^v m_v| m_v \in \tilde{\C}_i\right\}.}

Note that in the definition of $\tilde{\C}^l_i$, $\W$ {\bf corresponds to the maximum value of the channel output for the channel $H$, and not for the channel $\tilde{H}$}.  Decodability of the layered code $\tilde{\C}^l_i$ follows directly from Equation (\ref{y_y'}). Assuming that $\tilde{Y}$, is the received signal for channel $\tilde{H}$, we construct the output of channel $H$ using Equation (\ref{y_y'}). Thus, we are able to find the layered messages $M_v$ for $V\in\{0,1,...,l-1\}$, and reconstruct $\tilde{M}_v$ from them. Note that  $M_v$ can be decoded because of the choice of $\W$ as defined in \ref{layered_cb} for the code $\C$.

Next, we find the maximum value of the output for the channel $\tilde{H}$. Let $f_{max}$ ($f_{min}$) be the maximum (minimum) value of the vector $\vec{f}$. Maximum value of $\tilde{Y}$ can be upper-bounded as the following:

\eq{\label{y'<y}\max \tilde{Y} < f_{max} \max Y.}

We  also compute:
\begin{align}
\max Y_i &= \sum_{v=0}^{l-1} (\W)^v (\max \{Y_{i,v}\})= \sum_{v=0}^{l-1} (\W)^v W_i \notag\\
&< \W \sum_{v=0}^{l-1} (\W)^{v}=\W \frac{\W^l-1}{\W-1}< (\W)^{l+1}-1. \label{max_layered}
\end{align}

Since $\max Y=\displaystyle{\max_i ( \max Y_i)}$, we  conclude:

\eq{\label{maxy} \max Y < (\W)^{l+1}-1.}

Thus, combining Equations (\ref{y'<y}) and (\ref{maxy}), we get:
\eq{\label{maxy'}\max {\tilde{Y}} < f_{max} \cdot \left((\W)^{l+1}-1\right)<f_{max} (\W)^{l+1}.}

Now, we  compute the efficiency of $\tilde{\C_l'}$:
\[\ef(\tilde{\C}^l)=\frac{l \log\left(\displaystyle{\prod_{i=0}^{K-1} s_i}\right)}{ \max \tilde{Y}}> \frac{l \log\left(\displaystyle{\prod_{i=0}^{K-1} s_i}\right)}{ (l+1) \log(\W) + \log(f_{max})}.\]

Letting $l$ go to infinity, we get:

\eq{\label{upp_aef}\aef(\tilde{\C})\ge\ef(\C).}

In a manner similar to before, we lower bound \[\max \tilde{Y} > f_{min} \cdot ((\W)^l-1),\] and therefore:

\eq{\label{low_aef}\aef(\tilde{\C})\le\ef(\C).}

From Equations (\ref{upp_aef}) and (\ref{low_aef}), we have:

\[\aef(\tilde{\C})=\ef(\C).\]

This proves the theorem. 
\end{proof}
%

\subsection{Arithmetic Progression Codes}
In this subsection, we investigate achievable efficiency and the maximum  efficiency of a class of codes we refer to as ``arithmetic progression codes". 
We call a code an ``arithmetic progression code" when the codebook at each transmitter is an arithmetic progression set. It is formally stated as follows:

\begin{defin}
\label{arithprog}
 $\C=(\C_1,\C_2,...,\C_K)$ is an ``arithmetic progression code" if,  $\forall ~ i \in \{1,2,...,K\}$ we have:
\[\C_i= r_i\times \{0,1,...,s_i-1\}=r_i \Z_{s_i},\]
for some integers $s_i,r_i \ge 1$. We refer $r_i$'s as step size.
\end{defin}



The next lemma facilitates the finding of an arithmetic code for a channel $H$. Intuitively, the lemma states that, to find the best achievable efficiency for this channel, it is enough to check arithmetic progression codes with the unit step sizes.
We formally state this here.

\begin{lem}
\label{unit:stepsize}
Let $\C$ be a decodable arithmetic progression code with step size vector $\vec{r}=(r_1,r_2,...,r_K)$, set size vector $\vec{s}=(s_1,s_2,...,s_K)$, and efficiency $\ef(\C)$ for channel matrix $H$. There exists a channel matrix $H' \in \CC[H]$ and a decodable arithmetic progression code $\C '$ with step size vector $(1,1,...,1)$ and set size $\vec{s}$ where:
\[ \ef({\C}')=\ef(\C).\]
\end{lem}
\begin{proof}
Let $\vec{d}=(1,1,...,1)$. Consider $H'=\diag(\vec{d})H\diag(\vec{r})$. From the definition, $H' \sim H$. We define the following codebook for $H'$:
\[\C_i'=\{0,1,...,s_i-1\}.\]

Using a similar reasoning as used in proof of Theorem \ref{eq_class}, one can check that $\C_i '$ is decodable for $H'$. Also, from Equation (\ref{y_y'}), $\max Y'= \max Y$. Thus, the code $\C '$ has the same efficiency as $\ef(\C)$.
\end{proof}

In this section, we characterize the achievable rates when arithmetic progression codes with unit step size are used. Let  \[\gcd(H(i,j))_j=\gcd(H(i,1),H(i,2),...,H(i,K)),\] and
\[\gcd(H(i,j))_{j\neq i}=\gcd(H(i,1),...,H(i,i-1),H(i,i+1),...,H(i,K)).\]

Note that, $\gcd(H(i,j))_{j}$ always divides $\gcd(H(i,j))_{j\neq i}$.

Next theorem gives an essential characterization of the achievable efficiency for the arithmetic progression codes with unit step size.

\begin{thm}
\label{asym_thm}
Consider an integer-valued channel $H \in \Z^{K \times K}$. 
Let $s_i$ be defined as the following:
\[s_i=\frac{\gcd(H(i,j))_{j \neq i}}{\gcd(H(i,j))_{j }}.\]

The arithmetic progression code $\C=(\Z_{s_1},\Z_{s_2},...,\Z_{s_K})$ is decodable and achieves the following efficiency:

\[\ef(\C)=\frac{\log\left(\displaystyle{\prod_{i=1}^K s_i}\right)}{\log(\W)},\]

where $\W=\displaystyle{\max_i\{W_i\}+1}$, where $W_i$ is given by:

\begin{align*}
W_i= \sum_{j=1}^K H(i,j) \left( s_j-1\right).
\end{align*}
\end{thm}

Using this theorem, we can show the following achievable efficiency for the symmetric channels.  
\begin{cor}
For a symmetric integer-valued channel matrix, i.e., $H \in \Z^{K,K}$, $H(i,i)=a$ and $H(i,j)=h$ for all $i \neq j$, where $\gcd(a,h)=1$, there exists a code $\C$ with the following efficiency:
\[\ef(\C)=\frac{K \log(h)}{\log(h(a+(K-1)(h-1))+1-a)}.\] 
\end{cor}
\begin{proof}
In this case, one can check that $s_i=h$ for all $i\in\{1,2,...,K\}$. The corollary results  from Theorem \ref{asym_thm}.
\end{proof}

\begin{proof}[Proof of Theorem \ref{asym_thm}]
Consider an arithmetic progression code $\C_i=\Z_{s_i}$ as defined in the theorem.
One can check that, with this code-structure, we have  $\max{\Y_i}=W_i$. 
Next, we  show that this code is in fact decodable. Assume Transmitter $i$ desires to communicate the codeword $X_i=m_i \in \Z_{s_i}$.
Using a simple transformation (factoring), Receiver $i$ can rewrite $Y_i$ as:

\begin{align}
Y_i= \gcd&(H(i,j))_{j} \Big(\frac{H(i,i)}{\gcd(H(i,j))_{j }} m_i + \\ &\frac{\gcd(H(i,j))_{j \neq i}}{\gcd(H(i,j))_{j}}\sum_{j \neq i} \frac{H(i,j)}{\gcd(H(i,j))_{j \neq i}} m_j  \Big),
\end{align}
where all the fractions belong to the set of integers $\Z$.

Thus, utilizing the fact that 
\[\gcd\left(\frac{H(i,i)}{\gcd(H(i,j))_{j}},\frac{\gcd(H(i,j))_{j \neq i}}{\gcd(H(i,j))_{j}}\right)=1,\]

and therefore:
\[\exists \  h_i,  \mode{h_i \frac{H(i,i)}{\gcd(H(i,j))_{j}}}{\frac{\gcd(H(i,j))_{j \neq i}}{\gcd(H(i,j))_{j}}}=1. \]

Knowing that  $m_i \le s_i-1$, we can write:

\[m_i=\mode{\frac{h_i Y_i}{\gcd(H(i,j))_{j }}}{\frac{\gcd(H(i,j))_{j \neq i}}{\gcd(H(i,j))_{j}}}.\]
\end{proof}

\begin{defin}
Let $H \in \Z^{K \times K}$ be a channel matrix. Define $\C_H$ as the code obtained from Theorem \ref{asym_thm}, when $r_i=1$ for $i=1,2,...,L$.
\end{defin}

Next, we provide two examples that illustrate Theorems \ref{eq_class} and \ref{asym_thm}. This example shows how arithmetic progression codes can be used to construct layered codebooks with high efficiency.

\begin{ex}
Consider the following channel matrix:
\[H=\left[\begin{array}{ccc}
1 &a_1&a_2\\
a_3&1&a_4\\
a_5&a_6&1\end{array}
\right],\]
where $a_i>1$ and $\gcd(a_i,a_j)=1$ for all $i,j$. 

Let $H'$ be defined as the following:
\[H'=\left[\begin{array}{ccc}
a_4 a_6 &a_1 a_2 a_5&a_1 a_2 a_3\\
a_3 a_4 a_6&a_2 a_5&a_3 a_4 a_1\\
a_4 a_5 a_6&a_2 a_5 a_6&a_1 a_3\end{array}
\right],\]
One can check that $H\in \CC[H']$.

From Theorem \ref{asym_thm}, there exists a codebook $\C$ for $H'$ with the following efficiency: 
\[
\ef(\C) = \frac{\log(a_{1} a_{2} a_{3} a_{4} a_{5} a_{6} )}{\log(\W)}.
\]

Furthermore, we can upper bound $W_1$ as:
\[W_1 < a_{1} a_{2} \left ( a_{4} a_{6} + a_{3} a_{5} \left ( a_{4}+ a_{6} \right ) \right ).\] 
Next, we show that $\ef(\C)>1$. Without loss of generality, we assume that the maximum value for $\W$ is achieved for $i=1$; then we get:

\begin{align*}
\frac{\W}{a_1 a_{2} a_{4} a_{6}} &< 1 +a_{3} a_{5}( \frac{1}{a_{4}}+ \frac{1}{a_{6}}) \\
&<  a_{3} a_{5},
\end{align*}
or equivalently, $\W<a_1 a_2 a_3 a_4 a_5 a_6$ which means $\ef(\C)>1$. 

In this case, if all the $a_{i}$'s are of the same order and large, it can be shown, \cite[Corollary 3]{Jafarian2010}, that $\ef(\C) \lessapprox \frac{6}{5}$.

Now, using Theorem \ref{eq_class}, we know that there exists a layered codebook $\hat{\C}^l$ for the channel $H$, with the following asymptotic efficiency:
\[\aef(\hat{\C})=\ef(\C)>1.\]

Note that, if we apply Theorem \ref{asym_thm} for the original matrix $H$, we can not achieve an asymptotic efficiency of more than one. 
\end{ex}

Another example is given here.
\begin{ex}
Consider the following  channel matrix:
\[H=\left[\begin{array}{ccc}
1&a&b\\
a&1&c\\
b&c&1\end{array}
\right],\]
where $a,b$ and $c$ are pairwise coprime. Here, we show existence of  an asymptotic ``good" arithmetic progression code for this channel.

Let $H'$ be defined as:
\[H'=\left[\begin{array}{ccc}
c&a b&a b\\
a c&b& a c\\
 b c& b c& a\end{array}
\right].\]
One can check that $H\in\CC[H']$.
 From Theorem \ref{asym_thm}, there is a arithmetic progression codebook ${\C}$ for the channel $H'$ with the following efficiency:
\[\ef(\C)=\frac{2\log(a b c)}{\log(\W)},\]
where $\W<abc(\max\{a+b,a+c,b+c\}+1)$, which means $\ef(\C)>1$. Using the same steps as given in Theorem \ref{eq_class}, we can construct a layered codebook $\hat{\C}^l$ for the channel $H$, with the following asymptotic efficiency:
\[\aef(\hat{\C})=\ef(\C)>1.\] 
\end{ex}


%
%
%

Given this background, we state the main theorem of this section.
\begin{thm}
\label{thm_eq_ef}
For a given integer channel matrix $H$, let $\ef(\CC[H])$ be defined as in Equation (\ref{ef_class}). 
From the definition of $\ef(\CC[H])$ and Lemma \ref{unit:stepsize}, for any  $\epsilon>0$, there exists a matrix $H' \in \CC[H]$, such that $\ef(\C_{H'})>\ef(\CC[H])-\epsilon$.
We can construct a layered code $\tilde{\C}_l$ for channel matrix $H$, based on $\C_{H'}$ employing the construction proposed in Theorem \ref{eq_class} with the following asymptotic efficiency:

\[\aef(\tilde{\C})=\ef(C_{H'})> \ef(\CC[H])-\epsilon.\]
\end{thm}

\section{Back to Gaussian Interference Channel}
\label{sec:GIFC}


This section is divided into two parts. In the first subsection, we provide the overall achievable scheme corresponding to the sum-rate presented in Theorem \ref{integer_GIFC}. In the second subsection, we develop a modified version of this scheme to be used for non-integer channels that results in the achievable sum-rate  in Theorem \ref{thm_RGIFC}.

\subsection{Integer Channel Gains}
\label{sec:IGIFC}
In this section we present more discussions and proof for Theorem \ref{integer_GIFC}. 

In order to prove this theorem, we use nested lattices based on Construction-A to generate the codebooks, $\CG$ at each transmitter. Nested lattices $\LA{c} \subset \LA{f}$ are constructed as introduced in \cite{Erez2004} and stated in Equations (\ref{eqn:construct1})  and (\ref{eqn:construct2}). We present them again below:
\begin{align*}
\LA{f}&=\G_1 \left(\frac{1}{q}\G_2.\Z_q+\Z^n \right),\\
\LA{c}&=\G_1 \Z^n. 
\end{align*}

Let $\CG_0=\LA{f}\cap \V_c$. We employ Lemma 1  in \cite{Jafarian2009} to derive constraints on prime $q$ to  ensure  the existence of a ``good" matrix $\G_1$ and vector $\G_2$, where the notion of ``good'' is as defined in \cite{Erez2004}. We reproduce the relevant lemma below for convenience:

\begin{lem}[Lemma 1 in \cite{Jafarian2009}]
\label{lem:goodness}
Assuming that $q^\frac{2}{n}\le \frac{P}{Z}$, there exists a matrix $\G_1$ and a vector $\G_2$  such that the following hold true:
\begin{enumerate}
\item $\sigma^2(\V_c)=P$.
\item Probability of error in determining $\lambda \in \CG_0$ from $Y=\lambda+N$ (where $N$ is an AGN with variance $Z$) using lattice decoding can be made arbitrarily small for large lattice dimension $n$.
\end{enumerate}
\end{lem}

Let $H$ belongs to an equivalence class of $K\times K$ integer matrices $\CC$. Let $H'$ be the matrix in the same equivalence class $\CC$, from Theorem \ref{thm_eq_ef}. We know:
\[\ef(\C_{H'})>\ef(\CC)-\epsilon.\]

Let $H'=\diag(\vec{f})H \diag(\vec{t})$. Let $\W$ be the maximum channel output of the channel $H'$, when  Transmitter $i$ uses codebook ${\C_{H',i}}$. Therefore,

\[\ef(\C_{H'})=\frac{\log\left(\prod_{i=1}^K |\C_{H',i}|\right)}{\log(\W)}.\]

Consider the layered code $\tilde{\C}_l$ for channel $H$ which is constructed based on $\C$ from Theorem \ref{thm_eq_ef}. Let $\tilde{W}_{max,l}$ be the maximum output signal of the channel $HX$, when codebook $\tilde{\C}_l$ is used.

For any positive integer $n$, choose $l$ such that: 

\[2 \tilde{W}_{max,l}<\left(\frac{P}{Z}\right)^\frac{n}{2}<2 \tilde{W}_{max,l+1}.\]

Note that from Equations (\ref{maxy}) and (\ref{maxy'}), we can upper bound $\tilde{W}_{max,l+1}$  by $ \W \tilde{W}_{max,l}$. Thus,  we have:

\eq{\label{l,n}\frac{1}{n}\log(2 \tilde{W}_{max,l})<\frac{1}{2}\log\left(\frac{P}{Z}\right)<\frac{1}{n} \log(2 \W \tilde{W}_{max,l}).} 

Next, for this choice of $l$, we know there exists a  prime $q$ that satisfies:
\begin{equation}\label{cond:q}
 \tilde{W}_{max,l}<q<2 \tilde{W}_{max,l}.
\end{equation}

Note that $q<\left(\frac{P}{Z}\right)^\frac{2}{n}$ and therefore this choice of $q$ also satisfies the requirements imposed by Lemma \ref{lem:goodness}. 

%
%
%

%

We also know that $|\tilde{\C}_{i,l}|= |\C_{H',i}| ^l$.
With this background, we prove the main Theorem.
\begin{proof}[Proof of Theorem \ref{integer_GIFC}]

Let $\LA{f_i}=\G_1 \left( \frac{1}{q} \G_2 \tilde{\C}_{i,l} +\Z^n \right).$ Note that $\LA{f_i} \subset \LA{f}$ for all $i=1,2,3$. We also define $\CG_i = \LA{f_i} \cap \V_c$. It follows that $\CG_i\subset \CG_0$ and 
\begin{equation}
\label{ci:cardinality}
|\CG_i|=|\C_{i,l}|=|\C_{H',i}|^l.
\end{equation}

Consider the following encoding and decoding scheme for the Gaussian  interference channel:

{\bf Encoding Scheme:}
Transmitter $i$ chooses a codeword $X_i \in \CG_i$ associated with the desired message.

{\bf Decoding Scheme:}
Decoding is done in three steps. Each receiver first eliminates the additive Gaussian noise using lattice decoding as done in \cite{Loeliger1997}, then constructs a one-dimensional deterministic channel from the received lattice point. Next,  it determines the intended codeword from the resulting equivalent deterministic channel\footnote{Note that this notion is a special case of computational codes developed for more general network settings in \cite{Nazer2008}}. Let 
\[ Y_i=X_i+\sum_{j\neq i}H(i,j) X_j +N_i,\]
be the received signal at Receiver $i$ and denote
 \[
\tilde{Y}_i=Y_i-N_i=X_i+\sum_{j\neq i}H(i,j) X_j, 
\]
be the noise free, received signal.
\begin{enumerate}
\item {\bf AGN removal using lattice decoding:} Since channel coefficients are integers and $X_i\in \LA{f}$, $\tilde{Y}_i \in \LA{f}$. By choosing an appropriate prime $q$ (as given by Equation (\ref{cond:q})), we ensure that the transmit lattices are ``good'' for channel coding, and thus the noise $N_i$ can be ``removed'' from $Y_i$ to get $\tilde{Y}_i$ with high probability.

\item {\bf Construction of an equivalent one-dimensional deterministic channel:} Let 
\[X_i=\G_1 \left( \frac{1}{q} \G_2 e_i +{\bf z}_i\right),\]
where $s_i \in \tilde{\C}_{i,l}$ and ${\bf z}_i \in \Z^n$ (note that this assignment is unique). Without loss of generality, assume  that $g$, the first entry of the vector $\G_2$, is nonzero. Since $g$ is non zero and $g \in \Z_q$, $g$ has an inverse element in $\Z_q$, we call that $g^{-1}$. Define: 
\begin{equation}
\label{det:equ}
u_i \triangleq e_i+\sum_{j \neq i} H(i,j) e_j.
\end{equation} 

Note that $u_i$ is the output signal of a deterministic channel when each transmitter uses codebook $\tilde{\C}_{i,l}$. Thus from Equation (\ref{cond:q}), we have:
\[u_i<f_{max}\W^{l+1}<q.\]

From inequality stated above, one can check that:

\eq{\label{ui}u_i=[g^{-1}(q \G_1^{-1} \tilde{Y}_i)_1] \mod q, } where $({\bf v})_1$ is the first entry of a vector ${\bf v}$.

\item {\bf Determining $X_i$, and thus the intended message:} Using $u_i$ and decodability property of the layered code $\tilde{\C}_{i,l}$,  we can determine $e_i$. $X_i$ can be computed from $e_i$ as  follows:

\[X_i=\left[\frac{1}{q} \G_1 \G_2 e_i\right] \mod \LA{c}.\]
\end{enumerate}

To complete the proof of this theorem, we must determine the rate achieved through this scheme by each user.
Let $|\CG_i|=2^{nR_i}$. From Equations (\ref{ci:cardinality}, \ref{cond:W,P}) we get:

\[R_i=\frac{l}{n}\log\left(|\C_{H',i}|\right)\]

Corresponding sum-rate is:
\eq{\label{achsr}\sum_{i=1}^K R_i=\frac{l}{n}\log\left(\prod_{i=1}^K|\C_{H',i}|\right)=\frac{1}{n}\log(\tilde{W}_{max,l})\ef(\tilde{\C}_l)}

Now from Equation (\ref{l,n}), we can write:

\[\frac{1}{n}\log(\tilde{W}_{max,l}) +\frac{1}{n}\log(2)< \frac{1}{2}\log(P),\]
and
\[\frac{1}{2}\log\left(\frac{P}{Z}\right) < \frac{1}{n} \log( \tilde{W}_{max,l})+ \frac{1}{n} \log(2\W).\]

Therefore as $n$ becomes very large, and $l$ grows with respect to $n$ to satisfy Equation (\ref{l,n}), we have:

\begin{equation}
\label{cond:W,P}
\frac{1}{n}\log( \tilde{W}_{max,l})\approx \frac{1}{2}\log\left(\frac{P}{Z}\right).
\end{equation}

Combining Equations (\ref{achsr}) and (\ref{cond:W,P}), and letting $n$ goes to infinity, we desired result:
\[\sum_{i=1}^K R_i=\frac{1}{2}\log\left(\frac{P}{Z}\right)\aef(\tilde{\C})>\frac{1}{2}\log\left(\frac{P}{Z}\right)(\ef(\CC)-\epsilon),\]
where the last inequality follows from Theorem \ref{thm_eq_ef}.
\end{proof}

\subsection{Real Channel Gains}
\label{sec:RGIFC}
In the last section, a coding scheme proposed to achieve the sum-rate promised in Theorem \ref{integer_GIFC}. We observe that in general, this new scheme can achieve higher rates than time-share when the channel matrix is integer. Here, we investigate to find a modified scheme to achieve sum-rate proposed in the Theorem \ref{thm_RGIFC}, when channel matrix is real.

\begin{proof}[Proof of Theorem \ref{thm_RGIFC}]
In order to proof this theorem, each transmitter needs to use a random dither. Let $U \in \V{c}$ be a vector chosen uniformly from the Voronoi region of lattice $\LA{c}$. Let the codebooks $\CG_i$'s be defined as Theorem \ref{integer_GIFC}.

{\bf Encoding Scheme:}
Transmitter $i$ transmits a codeword $x_i$, where:
\[x_i=\mode{X_i-U}{\LA{c}},\]
and $X_i \in \CG_i$ associated with the desired message.

{\bf Decoding Scheme:} Decoding can be done similar to that of theorem \ref{integer_GIFC}. In the first two steps we want to remove the noise and reconstruct:
\eq{\label{dwgn}u_i=e_i+\sum_{j\neq i} \tilde{H}(i,j) e_j.}
\begin{enumerate}

\item {\bf AGN removal using lattice decoding:}
Let \[\tilde{Y}'_i=\displaystyle{\mode{\sum_{j=1}^K \tilde{H}(i,j)X_j}{\LA{c}}},\]
and 
\[\tilde{Y}_i=\displaystyle{\sum_{j=1}^K \tilde{H}(i,j)X_j},\]
as defined in the proof of Theorem \ref{integer_GIFC}.

In order to construct $\tilde{Y}'_i$, receiver $i$, first constructs the following:
\begin{align*}
\tilde{Y}_i&=\mode{Y_i+\sum_{j=1}^K \tilde{H}(i,j) U}{\LA{c}}\\
&=\mode{\sum_{j=1}^K \tilde{H}(i,j) (x_j+ U)+ \sum_{j=1}^K\Hd(i,j)x_j+N_i}{\LA{c}}\\
&=\mode{\sum_{j=1}^K \tilde{H}(i,j) \mode{x_j+ U}{\LA{c}}+ \sum_{j=1}^K\Hd(i,j)x_j+N_i}{\LA{c}}\\
&=\mode{\sum_{j=1}^K \tilde{H}(i,j) X_j+ \sum_{j=1}^K\Hd(i,j)x_j+N_i}{\LA{c}}\\
&=\mode{\sum_{j=1}^K \tilde{H}(i,j) X_j+ N'_i}{\LA{c}},
\end{align*}
where $N'_i=\displaystyle{\sum_{j=1}^K\Hd(i,j)x_j+N_i}$. From \cite[Lemma 1]{Erez2004}, we know $X_i$ and $x_i$ are independent for all $i$'s. Therefore, $N'_i$ and $\displaystyle{{\sum_{j=1}^K \tilde{H}(i,j) X_j}}$ are independent. Variance of the new noise $N'_i$ can be upper bounded as:
\[Z'\triangleq \sigma^2(N'_i) \le Z_{add}.\]

Now, we can choose $l$ and $q$ according to Equations (\ref{l,n}) and (\ref{cond:q}), respectively with replacing $Z'$ and $Z$.

With this choice of Lattices $\LA{f}$ and $\LA{c}$, receiver $I$ can decode $\tilde{Y}'_i$ from $\tilde{Y}$, using lattice decoding.

\item {\bf Construction of an equivalent one-dimensional deterministic channel:} Note that if we had $\tilde{Y}_i$, we could obtain $u_i$ in the similar way as stated in Equation (\ref{ui}). Although $\tilde{Y}'_i \neq \tilde{Y}_i$, we can write:

\[\tilde{Y}_i=\tilde{Y}'_i+\lambda,\]
where $\lambda \in \LA{c}$, i.e., $\lambda=\G_1 {\bf z}$ where ${\bf z} \in \Z^n$. If we rewrite the Equation (\ref{ui}) for $\tilde{Y}'_i$, instead of $\tilde{Y}_i$, we have:

\begin{align*}
\mode{g^{-1}(q \G_1^{-1} \tilde{Y}'_i)_1}{q}&=\mode{g^{-1}(q \G_1^{-1} (\tilde{Y}_i+\G_1 {\bf z}))}{q}\\
&=\mode{g^{-1}(q \G_1^{-1} \tilde{Y}_i)_1 + g^{-1} q z_1}{q}\\
&=\mode{g^{-1}(q \G_1^{-1} \tilde{Y}_i)_1}{q}\\
&=u_i.
\end{align*}

Thus, we can construct $u_i$, the same way using $\tilde{Y}'_i$.

\item {\bf Determining $x_i$, and thus the intended message:} We first decode $X_i$ from $u_i$ the same way as provided in step 3 of proof of Theorem \ref{integer_GIFC}. Next we compute $x_i$ from $X_i$ as the following:

\[x_i=\mode{X_i-U}{\LA{f}}.\] 

\end{enumerate}
This completes the proof.
\end{proof}
\section{Conclusion}
\label{sec:conc}
In this work, we develop achievable rates for the $K$-user Gaussian  interference channel. To accomplish this, we study equivalent discrete deterministic interference channels (DD-IFC) and then transform results obtained to the original G-IFC.  For the DD-IFC, we define a notion of ``efficiency" which measures the ``goodness" of the codes being constructed. We develop a new family of codes that attain a high efficiency and thus achieve non-trivial rates for the  original Gaussian IFC at finite SNRs. Although our initial analysis is for channels with integer coefficients, we extend our analysis to non-integral channels by utilizing dithered lattices.

\bibliographystyle{plain}
\bibliography{strongint}
\end{document}